\documentclass[12pt]{article}
\usepackage[margin=1in]{geometry}
\usepackage{amsbsy,amsmath,amsthm,amssymb,subfigure,graphics,graphicx,bm,setspace}
\usepackage[authoryear]{natbib}
\usepackage{changepage}
\RequirePackage[colorlinks,citecolor=blue,urlcolor=blue]{hyperref}

\usepackage{enumerate,enumitem}
 \usepackage{supertabular}
 
\newcommand{\comments}[1]{}
\newcommand{\pcite}[1]{\citeauthor{#1}'s \citeyearpar{#1}}

\newtheorem{lemma}{Lemma}
\newtheorem{definition}{Definition}
\newtheorem{theorem}{Theorem}

\newcommand{\sN}{{\cal N}}

\begin{document}

\title{Ridge partial correlation screening for ultrahigh-dimensional data}

\author{Run Wang, An Nguyen, Somak Dutta, and Vivekananda Roy\\
 Department of Statistics, Iowa State University, USA}

\date{}

\maketitle

\begin{abstract}

  Variable selection in ultrahigh-dimensional linear regression is
  challenging due to its high computational cost. Therefore, a
  screening step is usually conducted before variable selection to
  significantly reduce the dimension. Here we propose a novel and
  simple screening method based on ordering the absolute sample ridge
  partial correlations. The proposed method takes into account not
  only the ridge regularized estimates of the regression coefficients
  but also the ridge regularized partial variances of the predictor
  variables providing sure screening property without strong
  assumptions on the marginal correlations. Simulation study and a
  real data analysis show that the proposed method has a competitive
  performance compared with the existing screening procedures. A
  publicly available software implementing the proposed screening
  accompanies the article.
\end{abstract}


 \section{Introduction}
\label{sec:int}

In recent years, data sets with hundreds of thousands of variables are more and more common in diverse disciplines-- agronomy, biology, engineering, genomics, medicine, and the list goes on and on. Nevertheless, the number of variables which cause significant impact on the response are relatively small. For example, only a handful of genes are known to directly associate with Type 1 diabetes among a host of genes. Thus, variable selection methods have been attracting researchers' attention and in the past several years, we have witnessed a lot of research results appearing in this area. A widely used approach to perform variable selection and coefficient estimation simultaneously is by penalizing a loss function to get a shrinkage estimator of the coefficient. This class of methods include the popular lasso \citep{tibs:1996}, the elastic net
\citep{zou:hast:2005}, the SCAD \citep{fan:li:2001}, the adaptive lasso \citep{zou:2006}, the ridge regression \citep{hoer:kenn:1970} , and the bridge regression \citep{huan:horo:ma:2008}.

In ultrahigh-dimensional settings, the number of variables $p$ far exceeds the number of observations $n$. In such $p\gg n$ settings, the shrinkage methods may fail and the computational cost for the corresponding large-scale optimization problem is not ignorable. Thus, the feature screening methods, which aim to reduce the large dimensionality rapidly and efficiently, are performed in this setting before conducting a refined model selection analysis. Motivated by these, \cite{fan:lv:2008} introduced the sure independence screening (SIS) method in which the $n$ variables with the largest marginal (Pearson) correlations with the response survive after screening. The SIS method is then extended to generalized linear models \citep{fan:song:2010}, additive models \citep{fan:feng:song:2011}, and the proportional hazards model \citep{zhao:li:2012, gors:sche:2013}. Although SIS method can reduce the dimension rapidly to a manageable size, it is showed in \cite{wang:dutt:roy:2020} that SIS may fail in the presence of correlated predictors.
Based on \pcite{fan:lv:2008} work, several other correlation measures such as general correlation \citep{hall:mill:2009},
distance correlation \citep{li:zhon:zhu:2012}, rank correlation \citep{li:peng:zhan:zhu:2012}, tilted correlation \citep{cho:fryz:2012, lin:pang:2014} and quantile partial correlation \citep{ma:li:tsai:2017} have also been proposed as variable screening tools. Other variable screening methods include the traditional forward regression (FR) method \citep{wang:2009} \cite[see also][]{hao:zhan:2014} and high dimensional ordinary least squares projection (HOLP) \citep{wang:leng:2016}.

In this paper, we propose a novel variable screening method based on
the ridge partial correlation (RPC). This method is motivated by the
Bayesian iterated screening (BITS) of \cite{wang:dutt:roy:2021} who
introduced the ridge partial correlation to demonstrate contrasts of
BITS with FR. The proposed RPC screening method makes use of the
partial correlation information as well as the shrinkage effect due to
the ridge penalty. Given the response $y$ and the predictors
$(x_1, x_2, \dots, x_p)$, the partial correlation of $x_i$ and $y$ is
the correlation between $x_i$ and $y$ controlling the effect of all
other predictors $(x_1, \dots, x_{i-1}, x_{i+1}, x_p)$. This property
is inherited by RPC. Thus, the RPC screening will not likely select
the unimportant variables that are correlated with the important
variables, which is not the case with the SIS.  The implementation of
the RPC screening is straightforward and efficient by a parallel
computing algorithm proposed here. We prove that, asymptotically, RPC
screening has the same path as the HOLP and thus RPC has the sure
screening property. We show that, like the HOLP method, RPC has
screening consistency under the general $q$-exponential tail condition
on the errors and it does not require the marginal correlations for
the important variables to be bounded away from zero---an impractical
assumption used for showing screening consistency of SIS. Thus, RPC satisfies the two important aspects of a
successful screening method mentioned in \cite{wang:leng:2016}, namely
the computational efficiency and the screening consistency property
under flexible conditions.  Moreover,
although RPC has similarities with the HOLP, unlike HOLP, RPC takes
into account (ridge) partial variances of the variables, and numerical
examples show that RPC can lead to better performance than the HOLP
and other variable screening methods.

The rest of the paper is organized as follows. We describe the RPC
screening method in Section~\ref{sec:scree}. The screening consistency
property is discussed and established in
Section~\ref{sec:consistency}. In Section~\ref{sec:comfw}, we propose
a fast statistical computation algorithm for RPC
screening. Section~\ref{sec:simu} includes results from an extensive
simulation study. A real data set involving a mammalian eye disease
study is analyzed in Section~\ref{sec:data}. Section~\ref{sec:disc}
contains a short discussion. Several theoretical results, proofs of
the theorems and additional simulation results are given in the
supplement. The methodology proposed here is implemented in an
accompanying R package {\it rpc} \citep{r:rpc}.


 \section{Ridge partial correlation screening method}
\label{sec:scree}

\subsection{Linear regression model}

Consider the familiar linear regression model 
\begin{equation}
\label{eq:model}
   y =  \beta_0 + \beta_1 x_1 + \beta_2 x_2 + \ldots + \beta_p x_p + \epsilon
\end{equation}
where $y$ is the response, $x= (x_1, \dots, x_p)^{\top}$ is the random vector of predictors,  $\beta_0 \in \mathcal{R}$ is the intercept, $\beta = (\beta_1, \dots, \beta_p) \in\mathcal{R}^p$ is the vector of regression coefficients, and $\epsilon$ is the random error. We also write the model \eqref{eq:model} corresponding to $n$ realizations of $(x, y)$ as
\begin{equation}
  \label{eq:lm}
Y = \beta_0 1_n + X\beta + \varepsilon,  
\end{equation}
where $Y$ is the $n\times 1$ vector of response values, $1_n$ is the
$n\times1$ vector with all entries equal to 1, $X$ is the $n\times p$
design matrix, and $\varepsilon \in \mathcal{R}^n$ is the vector
consisting of identically and independently distributed residual
errors with mean zero and variance
$\sigma^2$. 
 
\subsection{The ridge partial correlation screening method}

Partial correlation measures the correlation between two variables
when conditioned on other variables. The marginal correlation based
screening methods like SIS \citep{fan:lv:2008} suffer from the problem
that unimportant variables those are marginally correlated with
important variables may get selected \citep{wang:dutt:roy:2020}. Replacing marginal correlation
with partial correlation a screening method helps to solve this
problem as partial correlation removes the effects of conditioned
variables. The (population) partial correlation between the $i$th
predictor variable $x_i$ and the response $y$ given the remaining
predictor variables is the negative of the $(1, i)$th entry of the
correlation matrix corresponding to $\tilde{\Sigma}^{-1}$ where
$\tilde{\Sigma}$ is the covariance matrix of the random vector $(y,
x)$. We assume that the $p$ columns of the matrix $X$ are centered and
scaled. If $p/n \to 0$, a consistent estimator of $\tilde{\Sigma}^{-1}$ is
given by
\begin{equation}
  \label{eq:covinv}
  n\begin{pmatrix}
\tilde{Y}^{\top}\tilde{Y} &\tilde{Y}^{\top}X\\
X^{\top}\tilde{Y} &X^{\top}X
\end{pmatrix}^{-1},
\end{equation}
where $\tilde{Y} = Y - \bar{Y}1_n,$ and $\bar{Y} = 1_n^{\top}Y/n$ is the average of the responses. However, if $p>n$, the inverse of the matrix in \eqref{eq:covinv} is not defined. This motivated us to consider the ridge regularized estimator of $\tilde{\Sigma}^{-1}$ given by $\bm{R}$, where
\begin{equation}
\label{eq:ridge}    
\bm{R} := n\begin{pmatrix}
\tilde{Y}^{\top}\tilde{Y} &\tilde{Y}^{\top}X\\
X^{\top}\tilde{Y} &X^{\top}X + \lambda I
\end{pmatrix}^{-1},
\end{equation}
and $I$ is the identity matrix. Note that, in \eqref{eq:ridge}, the predictor variables have been penalized while the response variable is not.
We now describe our proposed screening method. This method uses the sample ridge partial correlation between $x_i$ and $y$ as defined below.
\begin{definition}
\label{def:rpc}
For any $i$, the sample ridge partial correlation between $y$ and $x_i$ with ridge penalty $\lambda$ is given by
\begin{equation}
  \label{eq:rpc}
R_{i,\lambda} = -v_{iy, \lambda}/\{v_{i, \lambda}\times v_{y,
    \lambda}\}^{1/2},  
\end{equation}
 where $v_{iy, \lambda}$ is the $(1,i+1)$th
element of $\bm{R}$, $v_{i, \lambda}$ is the $(i+1, i+1)$th diagonal
element of $\bm{R}$, and $v_{y, \lambda}$ is the $(1,1)th$ element of
$\bm{R}$.
\end{definition}

For variable screening, we rank the absolute sample ridge partial
correlations and select the $K$ variables with the largest (absolute) ridge
partial correlations. More specifically, let $K$ be the number of
variables survived after screening. Then, given the ridge penalty $\lambda$, we
choose the submodel $\gamma^{K} \subseteq\{1,2,\ldots,p\}$ as
\[\gamma^{K} = \{i: |R_{i, \lambda}| \text{ is among the largest $K$ of all $|R_{i, \lambda}|$'s}\}.\]
Note that, \pcite{wang:leng:2016} HOLP procedure chooses the variables
based on the ridge estimator of $\beta_i$'s, that is, based on
$|v_{iy,\lambda}|$ values. Whereas, the ridge partial correlation also
takes into account the (ridge) partial variances $v_{i,\lambda}$'s of
the $i$th variable given the rest of the predictor variables and the
response. Thus, between two candidates with the same (absolute ridge)
regression coefficient estimate, the one with the smaller partial
variance is preferred. In different simulation examples in
Section~\ref{sec:simu}, we demonstrate that RPC can perform better
than HOLP.
\subsection{Screening consistency}

\label{sec:consistency}

In order to state the assumptions and the theoretical results, we use
the following notations. For two real sequences $\{a_n\}$ and
$\{b_n\}$, $a_n \sim b_n$ means ${a_n}/{b_n}\rightarrow c$ for some
constant $c>0$; $a_n \succeq b_n$ (or $b_n \preceq a_n$) means
$b_n = O(a_n)$; $a_n \succ b_n$ (or $b_n \prec a_n$) means
$b_n = o(a_n)$. Also, for any matrix $A$, let $\alpha_{min} (A)$ and
$\alpha_{max}(A)$ denote its minimum and maximum eigenvalues,
respectively, and let $\alpha_{min}^{*}(A)$ be its minimum nonzero
eigenvalue. Finally, for a positive definite matrix $A$, let
$ \text{cond}(A) = \alpha_{max}(A)/\alpha_{min} (A)$ denotes the condition number of $A$.



Recall that $X$ denotes the design matrix. Assume that the rows of $X$
are identically and independently distributed with zero mean vector
and covariance matrix $\Sigma$. That is, $\Sigma = $ cov $(x)$ is the
covariance matrix of the predictors. Define $Z = X\Sigma^{-1/2}$ and
$z = \Sigma^{-1/2} x$. For simplicity, as in \cite{wang:leng:2016}, we
assume that all diagonal entries of $\Sigma$ are equal to one, that
is, $\Sigma$ is the correlation matrix.

Performance of a screening method crucially depends on the tail
behavior of the errors $\varepsilon$. Screening consistency of HOLP
and BITS has been established in \cite{wang:leng:2016} and
\cite{wang:dutt:roy:2021}, respectively under the below mentioned general
$q-$exponential tail condition of the errors. 
\begin{definition}[$q$-exponential tail condition]
 \label{def:qexptail}
A zero-mean distribution $F$ is said to have $q$-exponential tail, if there exists a function $q:[0,\infty)\to \mathbb{R}$ such that for any $N\geq 1,$ $\eta_1,\ldots,\eta_N \stackrel{iid}{\sim} F$, $\ell\in{\cal R}^N$ with $\|\ell\| = 1$, and $\zeta > 0$ we have $P\left(\left|\ell^\top \eta\right| > \zeta\right)\leq\exp\{1-q(\zeta)\}$ where $\eta = (\eta_1,\ldots,\eta_N)^{\top}.$
\end{definition}
As shown in \cite{vers:2012}, for the standard normal distribution,
$q(\zeta) = \zeta^2/2,$ when $F$ is sub-Gaussian
$q(\zeta) = c_F\zeta^2$ for some constant $c_F > 0$ depending on $F,$
and when $F$ is sub-exponential distribution $q(\zeta) = c'_F\zeta$
for some constant $c'_F > 0$ depending on $F.$ Finally, if only first
$2k$ moments of $F$ are finite, then
$q(\zeta) - 2k\log(\zeta) = O(1).$ 

For a subset $\gamma \subset \{1, 2, \dots, p\}$, let $|\gamma|$ be
the cardinality of $\gamma$ and $X_\gamma$ be the $n \times |\gamma|$
submatrix of $X$ consisting of the columns of $X$ corresponding
$\gamma$. We make the following set of assumptions:
\begin{itemize}
    \item[A1]  $Y = \beta_01_n + X_t\beta_{t} + \varepsilon$ where $t$ is the true model for some subset $t \subset \{1,2,\dots,p\}$, $\varepsilon= (\varepsilon_1, \dots,\varepsilon_n),$ $\varepsilon_i/\sigma \stackrel{iid}{\sim} F_0$ which has $q$-exponential tail with unit standard deviation.
    \item[A2] $z$ has a spherically symmetric distribution and there are some constants $c_1>1$ and $C_1 > 0$ such that
    \[P\{\alpha_{\max}(p^{-1}ZZ^{\top}) > c_1 \text{ or }\alpha_{\min}(p^{-1}ZZ^{\top}) < c_1^{-1})\}\leq \exp (-C_1n).\]
    \item [A3] We assume that $\text{var}(y) = O(1)$ and that, for some $\kappa\geq 0, \nu\geq 0, \tau\geq 0$ and 
    $c_2, c_3, c_4>0$ independent of $n , p$, 
    \[\min_{j\in t}|\beta_j|\geq \dfrac{c_2}{n^{\kappa}}, \qquad |t|\leq c_3n^{\nu}\quad\text{and}\quad \text{cond}(\Sigma)\leq c_4n^{\tau}.\]

    \item [A4] Finally, we assume 
    \[p\succ n^{1+\tau}, \quad \log(p) \prec \min\left\{\dfrac{n^{1-2\kappa-5\tau}}{2\log (n)}, q\left(\dfrac{\sqrt{C_1}n^{1/2-2\tau-\kappa}}{\sqrt{\log(n)}}\right)\right\}\]
    and $\lambda \prec n^{1-(5/2)\tau-\kappa}$.
    \end{itemize}
    The conditions A1--A3 as well as the second condition on $p$ in the assumption
    A4 are also assumed in \cite{wang:leng:2016} for proving screening consistency of the HOLP. 
    Let
    $S_\lambda = \tilde{Y}^{\top}\tilde{Y} -
    \tilde{Y}^{\top}X(X^{\top}X+\lambda I)^{-1}X^{\top}\tilde{Y} = n/v_{y, \lambda}$ 
    be the ridge residual sum of squares. The screening consistency
    result for the RPC method is stated in the following theorem.
\begin{theorem}
\label{thm:consistency}
Under the assumptions {\normalfont A1--A4}, there exists $\rho_n$ (defined in Appendix~\ref{app:consistency}) such that
\begin{equation*}
    \begin{split}
        P\bigg(\min_{i\in t} R_{i,\lambda}^2 > \dfrac{\rho_n^2}{S_\lambda + \rho_n^2} > \max_{i\notin t} R_{i,\lambda}^2\bigg)
        &= 1 - O\bigg(\exp\bigg\{-C_1\dfrac{n^{1-5\tau-2\kappa-\nu}}{2\log(n)}\bigg\}\\& +
  \exp\bigg[1-\dfrac{1}{2}q\bigg\{\dfrac{\sqrt{C_1}n^{1/2-2\tau-\kappa}}{\sqrt{\log(n)}}\bigg\}\bigg]\bigg).
        \end{split}
\end{equation*}

Alternatively, there exist $K_n \sim n^{\iota}$ for some $\iota\in(\nu,1]$ such that 
\[P(t\subset \gamma^{K_n}) = 1- O\left(\exp\left\{-C_1\dfrac{n^{1-2\kappa-5\tau-\nu}}{2\log(n)}\right\} + \exp\left[1-\dfrac{1}{2}q\left(\dfrac{\sqrt{C_1}n^{1/2-2\tau-\kappa}}{\sqrt{\log(n)}}\right)\right]\right)\]
\end{theorem}
\begin{proof}
The proof of the theorem is given in Section~\ref{app:consistency} of the supplement.
\end{proof}

The second result, which directly follows
from the first result, reveals that, if a submodel of size larger than
the true model is selected, then the submodel contains the true model with probability tending to one.


 \section{Fast statistical computation }
\label{sec:comfw}

In this section we describe how the ridge partial correlation coefficients are computed at the same order of computing cost as HOLP. To that end, let $W = XX^\top + \lambda I.$ Then, using Woodbury matrix identity,
\begin{equation}
\begin{split}
\label{comp:vy}
    v_{y, \lambda}/n &= (\tilde{y}^{\top}\tilde{y} - \tilde{y}^{\top}X(X^{\top}X+\lambda I)^{-1}X^{\top}\tilde{y})^{-1}\\
                   &= (\tilde{y}^{\top}\tilde{y} - \tilde{y}^{\top}X\left\{\dfrac{1}{\lambda}I-\dfrac{1}{\lambda}X^{\top}(XX^{\top}+\lambda I)^{-1}X\right\}X^{\top}\tilde{y})^{-1} \\
                   &= (\tilde{y}^{\top}\tilde{y} - \dfrac{1}{\lambda}\tilde{y}^{\top}XX^{\top}\tilde{y} + \dfrac{1}{\lambda}\tilde{y}^{\top}XX^{\top}(XX^{\top}+\lambda I)^{-1}XX^{\top}\tilde{y})^{-1}\\
                   &= (\tilde{y}^{\top}\tilde{y} - \dfrac{1}{\lambda}\tilde{y}^{\top}XX^{\top}\tilde{y} + \dfrac{1}{\lambda}\tilde{y}^{\top}XX^{\top}\tilde{y} - \tilde{y}^{\top}(XX^{\top}+\lambda I)^{-1}XX^{\top}\tilde{y})^{-1}\\
                   &= (\tilde{y}^{\top}\tilde{y} - \tilde{y}^{\top}\tilde{y} + \lambda \tilde{y}^{\top}(XX^{\top}+\lambda I)^{-1}\tilde{y})^{-1}\\
                   &= 1/\left[\lambda \tilde{y}^{\top}W^{-1}\tilde{y}\right].
\end{split}                   
\end{equation}
Similarly, if $e_i$ denotes the $i$th canonical basis vector of $\mathbb{R}^p,$ then,
\begin{equation} 
\label{comp:viy}
    \begin{split}
       v_{iy, \lambda}/n &= -e_i^\top(X^{\top}X + \lambda I)^{-1}X^{\top}\tilde{y}~/~ \left[\tilde{y}^{\top}\tilde{y} - \tilde{y}^{\top}X(X^{\top}X+\lambda I)^{-1}X^{\top}\tilde{y}\right]\\
                       &= -X^{\top}_iW^{-1}\tilde{y}/ \left[ \lambda \tilde{y}^{\top}W^{-1}\tilde{y}\right],
    \end{split}
\end{equation}
and finally,
\begin{equation}
\label{comp:vi}
    \begin{split}
        v_{i, \lambda}/n &= e_i^\top\left\{(X^{\top}X + \lambda I)^{-1} + \dfrac{(X^{\top}X + \lambda I)^{-1}X^{\top}\tilde{y}\tilde{y}^{\top}X (X^{\top}X + \lambda I)^{-1}}{\tilde{y}^{\top}\tilde{y} - \tilde{y}^{\top}X(X^{\top}X+\lambda I)^{-1}X^{\top}\tilde{y}}\right\}e_i\\
                       &= e_i^\top\left\{\dfrac{1}{\lambda}I-\dfrac{1}{\lambda}X^{\top}(XX^{\top}+\lambda I)^{-1}X\right\}e_i + \dfrac{(X^{\top}_iW^{-1}\tilde{y})^2}{\lambda \tilde{y}^{\top}W^{-1}\tilde{y}}\\
                       &= \dfrac{1}{\lambda} + \dfrac{(X^{\top}_iW^{-1}\tilde{y})^2}{\lambda \tilde{y}^{\top}W^{-1}\tilde{y}} - \dfrac{1}{\lambda}X^{\top}_iW^{-1}X_i
    \end{split}
\end{equation}

\subsection{Algorithm for computing ridge partial correlations}
In practice, these quantities can be efficiently computed using the Cholesky factor of $W$: Let $S$ be an upper triangular matrix such that $W = S^{\top}S$. Let $\theta = S^{-\top}\tilde{y}, u_i = S^{-\top}X_i$. Then $\tilde{y}^{\top}W^{-1}\tilde{y} = \theta^{\top}\theta, X^{\top}_iW^{-1}\tilde{y} = u^{\top}_i\theta$ and $X^{\top}_iW^{-1}X_i = u^{\top}_iu_i$. Substitute these equations into \eqref{comp:vy}, \eqref{comp:viy} and \eqref{comp:vi}, we get the following algorithm for computing ridge partial correlations:
\begin{enumerate}
    \item Given $X, \tilde{y}$ and $\lambda$, compute $W = XX^{\top} + \lambda I$.
    \item Compute the Cholesky factor of $W$, denoted by $S$, which is upper triangular satisfying $W = S^{\top}S$.
    \item Compute $\theta = S^{-\top}\tilde{y}$ and $v_{y, \lambda} = n/(\lambda \theta^{\top}\theta).$
    \item For each $i = 1,2,\ldots,p$, 
    \begin{itemize}
     \item[--] compute $u_i = S^{-\top}X_i$, $v_{iy, \lambda} = -v_{y,\lambda}\cdot u_i^{\top}\theta$ and 
     $v_{i, \lambda} = n\lambda^{-1} + v_{y, \lambda}(u_i^{\top}\theta)^2 - n\lambda^{-1}u^{\top}_iu_i.$
     \item[--] Compute $R_{i,\lambda} = -v_{iy, \lambda}/\{v_{i, \lambda}\times v_{y, \lambda}\}^{1/2}$
    \end{itemize}
\end{enumerate}

In the above algorithm, the third step can be parallelized over $i$ as implemented in the R package \emph{rpc} \citep{r:rpc}. Overall, the computational complexity of the above algorithm is $O(n^3) + O(n^2p)$ where the $O(n^3)$ term correspond to computing the Cholesky factorization of $XX^\top + \lambda I_n$ and the $O(n^2p)$ term arises due to computing $XX^\top$ and for computing $u_i$s. This is encouragingly the same as the computational complexity of HOLP.

 \section{Simulation study}
\label{sec:simu}

In this section, we compare the performance of the proposed ridge partial correlation screening with SIS (\cite{fan:lv:2008}), forward
regression (\cite{wang:2009}), BITS (\cite{wang:dutt:roy:2021}) and HOLP (\cite{wang:leng:2016}) by extensive simulation studies.

\subsection{Simulation settings}

We consider seven data generating models in our numerical studies. For
these settings E.1–E.7, the rows of $X$ are generated from
multivariate normal distributions with mean zero and different
covariance structures. We consider $n = 300$ and $p=5000$ in all
simulation settings. The response variable $y$ is generated by model
\eqref{eq:model} with error term $\epsilon$ following three separate
distributions with different tail behaviors. In the first case, we consider $\epsilon$ to be normally distributed with
mean 0 and variance 1. Next, we take $\epsilon$ to be shifted exponentially distributed with mean
0 and variance 1 and support $[-1, \infty)$, that is,
$\epsilon \stackrel{d}{=}\zeta - 1$ where $\zeta \sim \text{Exp}(1)$. Finally, we consider $\epsilon$ to be
scaled student $t$ distributed with degrees of freedom 20, mean 0
and variance 1, that is,
$\epsilon \stackrel{d}{=}\eta / \sqrt{20/18}$ where
$\eta \sim t_{20}$. Four values of theoretical $R^2 (=$ Var $(x^\top \beta)/$ Var $(y))$, namely
$R^2= 0.1,0.2,0.3,$ and $0.4$ are assumed. In E.1 - E.6, $\beta_j = 1$ for $j\leq 9$
and $\beta_j = 0$ for $j>9$. In E.7, $\beta_j = 1$ for $j\leq 25$ and
$\beta_j = 0$ for $j>25$.
\begin{enumerate}[label=E.\arabic*,leftmargin=*]
    \item \label{sec:ind} \textbf{Independent predictors} ~~ In this example, $\Sigma= I_p,$ the identity matrix of dimension $p$.
    \item \label{sec:comp}\textbf{Compound symmetry} ~~ Here, $\Sigma = \rho 1_p1_p^\top + (1-\rho)I_p.$ We set $\rho = 0.5$ in our study.
    \item \label{sec:auto} \textbf{Autoregressive correlation} ~~ This correlation structure is
appropriate if there is an ordering (say, based on time) in the
covariates, and variables further apart are less correlated. We use the AR(1) structure where the $(i,j)$th entry of $\Sigma$ is $ \rho^{|i-j|}.$ We set $\rho = 0.5$ in our study.
    \item \label{sec:fac} \textbf{Factor models} ~~ This example is from \cite{wang:leng:2016}. Fix $k=10.$ Let $F$ be a $p\times k$ random matrix with entries iid $\mathcal{N}(0,1).$ The $\Sigma$ is then given by $FF^\top + I_p.$
    \item \label{sec:gr} \textbf{Group structure} ~~ In this example, the true variables in the same group are highly correlated. This is a modification of example 4 of \cite{zou:hast:2005}. The predictors are generated by 
$X_{m} = z_1 + \zeta_{1, m}$, $X_{3+m} = z_2 + \zeta_{2, m}$,
$X_{6+m} = z_3 + \zeta_{3, m}$  where
$z_i \overset{iid}\sim \sN_n(0, I_n)$,
$\zeta_{i, m} \overset{iid}\sim \sN_n(0, 0.01I_n)$ and $z_i$'s and $\zeta_{i, m}$'s are independent for $1\leq i \leq 3$ and for $m=1,2,3$.
    \item \label{sec:ext} \textbf{Extreme correlation} ~~ This example is a modification of the challenging example 4 of \cite{wang:2009}, making it more complex. Generate $Z_i \stackrel{iid}{\sim} {\cal N}_{n} (0, I), i= 1, \dots, p$,
and $W_i \stackrel{iid}{\sim} {\cal N}_{n} (0, I), i= 1, \dots,9$. Set
$X_i = (Z_i + W_i)/\sqrt{2}, i= 1, \dots, 9$ and
$X_i = (Z_i + \sum_{i=1}^{9}W_i)/2$ for $i= 10, \dots, p$. The marginal correlation between $y$ and $X_i, i \geq 10$ (unimportant variables) is $2.5/\sqrt{3}$ times larger than the same between $y$ and $X_j, j\leq 9$ (important variables).

\item \label{sec:sparsefac} \textbf{Sparse factor models} ~~ This is the sparse version of \ref{sec:fac}. For each $1\leq j\leq 5,$ generate $f_{ij} \sim \sN(0,1)$ if $5(j-1)+1\le i \le 5j,$ and $f_{ij} = 0$ otherwise. Also, $\Sigma = FF^\top + 0.01I_p.$
\end{enumerate}

For each simulation setup, a total of 100 data sets are generated and we evaluate the performance of each method by the coverage probability (CP) and true positive rate (TPR). The CP and TPR are defined as follows. Let $\hat{\gamma}^{(i)}$ denote the chosen model in the $i$th replication for $i = 1,2,\ldots, 100.$ The coverage probability is calculated by CP = $\sum_{i=1}^{100} \mathbb{I}(\widehat{\gamma}^{(i)} \supseteq t)/100$. The true positive rate is calculated by TPR =  $(1/100)\sum_{i=1}^{100} \big|\widehat{\gamma}^{(i)}\cap t\big|/|t|.$ The CP measures whether all the important variables are survived after screening while TPR measures how many important variables are survived after screening.

\subsection{Summary of simulation results}

The simulation results for $R^2 = 0.2$ and normally distributed errors are shown in Table~\ref{tab:simulation}. Three distinct values of $\lambda$, namely $p/n$ (RPC1), $n \log n/p$ (RPC2) and $n/p$ (RPC3) are considered in our study. All these three methods select a submodel of size $n$. The URPC selects the union of submodels given by RPC1, RPC2 and RPC3. The UBITS selects the union of submodels of size $n$ given by BITS with the same three choices of $\lambda$ as RPC1, RPC2, and RPC3, respectively. 

In general, the RPC method is not sensitive to the value of $\lambda$. In most cases, the RPC method leads to similar performance as the HOLP and is competitive with the other methods. However, in the most difficult case of extreme correlation~\ref{sec:ext}, RPC beats all other methods, including the HOLP. Surprisingly, the RPC methods have almost twice coverage probability than the HOLP in the setup \ref{sec:ext}. Therefore, despite having some similarities with the HOLP, RPC provides better performance than the HOLP.

Results from the other simulation settings corresponding to different
$R^2$ values, and other error distributions are provided in Section~\ref{app:simu} of the supplement.
\begin{table}
\begin{center}
\caption{Simulation results (in \%) for $R^2 = 0.2$ and error normally distributed.}
\label{tab:simulation}
\begin{small}
\begin{adjustwidth}{-.05cm}{}
\begin{tabular}{rrr|rr|rr|rr|rr|rr|rr}
\hline\hline\multicolumn{1}{c}{Method}    & \multicolumn{12}{c}{Correlation Structure}\\

\hline
                        & \multicolumn{2}{c|}{IID}      & \multicolumn{2}{c|}{Compound} & \multicolumn{2}{c|}{Group}    & \multicolumn{2}{c|}{AR}       & \multicolumn{2}{c|}{Factor}  & \multicolumn{2}{c|}{ExtrCor}  &
                        \multicolumn{2}{c}{SpFactor}\\
                           & TPR  & CP  & TPR  & CP  & TPR  & CP  & TPR  & CP  & TPR  & CP  & TPR  & CP  &TPR  &CP  \\\hline
\multicolumn{1}{l}{RPC1}  
                               & 75.6 & 9 & 16.8 & 0 & 98.3 & 92 & 95.6 & 69 & 14.8 & 0 & 82.8 & 17 & 55.9 & 0    \\

\multicolumn{1}{l}{RPC2}  
                               & 75.6 & 9 & 16.8 & 0 & 98.3 & 92 & 95.6 & 69 & 15.1 & 0 & 82.8 & 19 & 56.0 & 0    \\

\multicolumn{1}{l}{RPC3}  
                              & 75.6 & 9 & 16.8 & 0 & 98.3 & 92 & 95.6 & 69 & 14.8 & 0 & 82.8 & 17 & 55.9 & 0    \\

\multicolumn{1}{l}{URPC}  
                               & 75.6 & 9 & 16.8 & 0 & 98.3 & 92 & 95.6 & 69 & 15.1 & 0 & 83.0 & 19 & 56.0 & 0    \\

\multicolumn{1}{l}{UBITS}  
                              & 74.3 & 9 & 22.8 & 0 & 98.3 & 94 & 91.7 & 44 & 23.0 & 0 & 35.7 & 0    & 70.0 & 1 \\

\multicolumn{1}{l}{HOLP}  
                              & 75.6 & 9 & 16.9 & 0 & 98.3 & 92 & 95.6 & 69 & 14.9 & 0 & 77.4 & 10  & 56.0 & 0    \\

\multicolumn{1}{l}{FR}  
                              & 32.1 & 0 & 7.7 & 0 & 34.2 & 0 & 34.1 & 0 & 7.3 & 0 & 12.2 & 0    & 13.6 & 0    \\

\multicolumn{1}{l}{SIS}   
                              & 77.4 & 11 & 23.8 & 0 & 99.6 & 98 & 97.3 & 80 & 20.1 & 0 & 0.4 & 0    & 62.5 & 1\\\hline
\end{tabular}
\end{adjustwidth}


\end{small}
\end{center}
\end{table}


 \section{Real data application}

\label{sec:data}

\cite{Bai:2019SpikeandSlabGL} analyzed a mammalian eye disease data set where gene expression measurements on the eye tissues from 120 female rats were recorded. This data set was originally analyzed by \cite{Scheetz:2006}. In this data set, the gene TRIM32 is considered to be responsible for causing Bardet-Biedl syndrome and the goal is to find out genes which are correlated with TRIM32. The original data set consists of 31,099 probe sets. Following  \cite{Bai:2019SpikeandSlabGL}, we chose the top 25000 probe sets with the highest sample variance (in logarithm scale). 

In this study, a cross validation was applied. Specifically, in each replication, we randomly split the whole data set into a training set of size $n = 100$ and a testing set of size 20. Submodels of size $n$ were selected by RPC with different $\lambda$ values: $p/n$ (RPC1), $n\log n/p$ (RPC2) and $n/p$ (RPC3). We also considered the union of these submodels. We compared these methods with BITS, HOLP, FR, and SIS. For BITS, we used the union of submodels given by BITS with the same $\lambda$ values as the RPC. 
After the screening, two model selection methods, SCAD and Lasso, were applied to the submodels to get the final selected models. 

From the box plots in Figure~\ref{fig:realdata}, we see that the RPC screening followed by the Lasso model selection method yields more accurate predictions than other methods. Results from the union of the RPC methods do not show too much difference with the results from the individual RPC methods, both in terms of prediction accuracy and the final model size. The reason might be that the RPC screening is insensitive to the choice of the hyperparameter $\lambda$ and RPC1, RPC2 and RPC3 produce similar models. 

\begin{figure}[htbp]
\begin{center}
  \includegraphics[width = 0.85\textwidth]{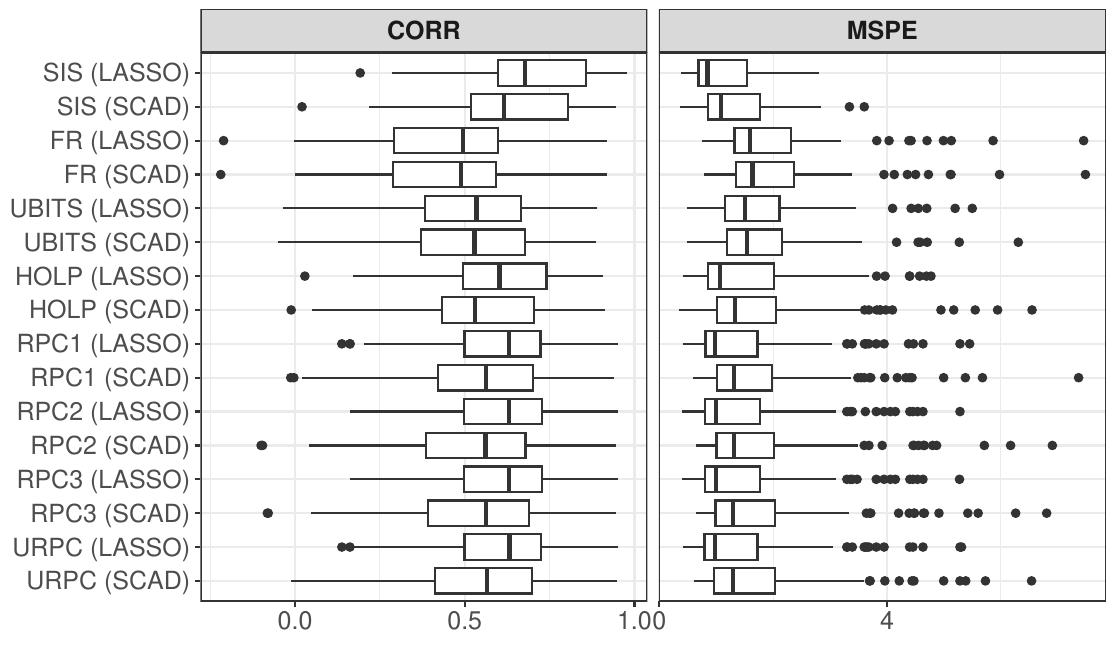}
\begin{tabular}{c|cc|cc}
  \hline
            & \multicolumn{2}{c|}{SCAD} & \multicolumn{2}{c}{LASSO} \\
          & Mean        & SE         & Mean        & SE          \\\hline
SIS & 7.91 & 3.54 & 16.60 & 2.22 \\ 
FR & 50.66 & 1.66 & 62.34 & 4.70 \\ 
UBITS & 32.42 & 13.60 & 37.33 & 15.55 \\ 
HOLP & 23.33 & 4.73 & 38.23 & 5.24 \\ 
RPC1 & 26.65 & 3.48 & 43.45 & 5.93 \\ 
RPC2 & 26.80 & 3.70 & 43.39 & 5.76 \\ 
RPC3 & 26.46 & 3.88 & 43.41 & 5.86 \\ 
URPC & 26.75 & 4.17 & 43.73 & 6.10 \\ 
   \hline
\end{tabular}


\end{center}
    \caption{Prediction accuracy on the test sets: correlation (top left), MSPE ($\times 10^{-4}$, top right). The mean model sizes and their standard errors are shown in the table.} 
     \label{fig:realdata}
\end{figure}

 \section{Discussion}
\label{sec:disc}

In this paper, we develop a ridge partial correlation (RPC) method for
screening variables in an ultrahigh-dimensional regression model. RPC
screening does not require strong assumptions on marginal
correlations, which means that it can be applied to more data sets with
complex correlation structures than the SIS. Theoretically, as
$n,p\to\infty$, the RPC screening will produce the same model as the HOLP, which
ensures that the performance of the RPC screening is not worse than the HOLP. We
show that the RPC screening leads to better performance in the extreme
correlation cases where the marginal correlations between the response
and unimportant variables are larger than the same between the
response and the true predictors. Finally, we show that the RPC
screening can be fast implemented by using Cholesky decomposition and
parallel computing.

The definition of ridge partial correlation does not require the linear
regression assumption. RPC screening can also be extended to larger
classes of regression models, for example the generalized linear
regression models. Potential future works include developing RPC
screening methods for generalized additive models.

\section*{Acknowledgement}
 The work was partially supported by USDA NIFA 2023-70412-41087 (S.D., V.R.), USDA Hatch Project IOW03717 (S.D.), and the Plant Science Institute (S.D.).

\begingroup
\small
\bibliographystyle{asadoi}
\bibliography{ref}
\endgroup

\newpage
\begin{center}
\textbf{\large Supplement to \\
``Ridge partial correlation screening for ultrahigh-dimensional data" \\Run Wang, An Nguyen, Somak Dutta, and Vivekananda Roy}
\end{center}
\setcounter{equation}{0}
\setcounter{figure}{0}
\setcounter{table}{0}
\setcounter{page}{1}
\setcounter{section}{0}
\makeatletter
\renewcommand{\thesection}{S\arabic{section}}
\renewcommand{\thesubsection}{\thesection.\arabic{subsection}}
\renewcommand{\theequation}{S\arabic{equation}}
\renewcommand{\thefigure}{S\arabic{figure}}
\renewcommand{\bibnumfmt}[1]{[S#1]}
\renewcommand{\citenumfont}[1]{S#1}
\renewcommand{\thetable}{S\arabic{table}}
 \section{Some notations}
  \label{app:notation}
\label{sec:notations}
The following notations will be used in the proof of Theorem~\ref{thm:consistency}:
\begin{itemize}
    \item $\mathcal{O}(n)$ denotes the $n$th order orthogonal group.
    \item The Stiefel manifold $V_{n,p}$ is the set of $p \times n$ matrices $A$ such that $A^\top A = I_n$ \citep{Chikuse:2003}.
\end{itemize}
For any $\lambda > 0$, 
\begin{itemize}
    \item $\hat{\beta}(\lambda) = (X^{\top}X + \lambda I)^{-1}X^{\top}\tilde{Y} = X^{\top}(XX^{\top}+\lambda I)^{-1}\tilde{Y}$ is the ridge regression estimator of $\beta$.
    \item $Z = VDU^{\top}$ is the singular value decomposition of $Z$, where $V\in\mathcal{O}(n)$,  $U\in V_{n,p}$, and $D$ is an $n \times n$ diagonal matrix. Thus, $X = VDU^{\top}\Sigma^{1/2}$.
    \item $A = (U^{\top}\Sigma U)^{1/2}$ satisfies $A^{\top}A = AA^{\top}=A^2$.
    \item $H = \Sigma^{1/2}U(U^{\top}\Sigma U)^{-1/2}$ satisfies $H^{\top}H = I_n$.
    \item $M = \sum_{k=1}^{\infty}\lambda^k\Sigma^{1/2}UA^{-1}(A^{-\top}D^{-2}A^{-1})^kA^{-\top}U^{\top}\Sigma^{1/2}$.
    \item $\xi_i$ is the $i$th diagonal element of $(X^{\top}X + \lambda I)^{-1}, i=1,\dots,p$.
\end{itemize}

\section{Some useful lemmas}
\label{sec:lemmas}
In this Section, we present some useful lemmas which will be used in the proof of Theorem \ref{thm:consistency}.

\begin{lemma}
\label{lemma:rpcrewrite}
In Definition \ref{def:rpc}, $v_{iy, \lambda}$ is the same as the $i$th element of 
\[\dfrac{-(X^{\top}X + \lambda I)^{-1}X^{\top}\tilde{Y}}{S_\lambda} ,\]
 
 $v_{i, \lambda}$ is the same as $i$th diagonal element of 
 \[(X^{\top}X + \lambda I)^{-1} + \dfrac{(X^{\top}X + \lambda I)^{-1}X^{\top}\tilde{Y}\tilde{Y}^{\top}X (X^{\top}X + \lambda I)^{-1}}{S_\lambda},\]
 and 
 \[v_{y, \lambda} = \dfrac{1}{S_\lambda}.\]

\end{lemma}

\begin{proof}
By block matrix inverse formula, 
\begin{align*}
    \bm{R} &= n\begin{pmatrix}
\tilde{Y}^{\top}\tilde{Y} &\tilde{Y}^{\top}X\\
X^{\top}\tilde{Y} &X^{\top}X + \lambda I
\end{pmatrix}^{-1}\\
&=n\begin{pmatrix}
\bm{A_\lambda} &\bm{B_\lambda}\\
\bm{C_\lambda} &\bm{D_\lambda}
\end{pmatrix}
\end{align*}

where
\begin{equation*}
    \begin{split}
        &\bm{A_\lambda} = \dfrac{1}{S_\lambda},\\
        &\bm{B_\lambda} = -\dfrac{\tilde{Y}^{\top}X(X^{\top}X + \lambda I)^{-1}}{S_\lambda},\\
        &\bm{C_\lambda} = -\dfrac{(X^{\top}X + \lambda I)^{-1}X^{\top}\tilde{Y}}{S_\lambda}, \; \mbox{and}\\
        &\bm{D_\lambda} = (X^{\top}X + \lambda I)^{-1} + \dfrac{(X^{\top}X + \lambda I)^{-1}X^{\top}\tilde{Y}\tilde{Y}^{\top}X (X^{\top}X + \lambda I)^{-1}}{S_\lambda}.
    \end{split}
\end{equation*}

Then the proof follows from Definition \ref{def:rpc} and an application of the standard block matrix inversion formula. 
\end{proof}

The following lemma is a well known result in linear algebra. 
\begin{lemma}
\label{lemma:sumofeigen}
For any $n\times n$ positive semidefinite matrices $A$ and $B$, we have $\alpha_{max}(A+B)\leq \alpha_{max}(A) + \alpha_{max}(B), \alpha_{max}(AB)\leq \alpha_{max}(A)\alpha_{max}(B), \alpha_{min}(AB)\geq \alpha_{min}(A)\alpha_{min}(B).$
\end{lemma}

\begin{proof}
A proof can be found in \cite{johnson:1985}.
\end{proof}

The following lemma is the Lemma 4 in \cite{wang:leng:2016}.
\begin{lemma}[\cite{wang:leng:2016}]
\label{lemma:Hnormbound}
Under assumptions A1 - A3, for any $C > 0$ and any fixed vector $v$ with $\|v\| = 1$, there exist $c_1', c_2'$ only related to $C$ with $0<c_1'<1<c_2'$ such that
\[P\left(v^{\top}HH^{\top}v < c_1'\dfrac{n^{1-\tau}}{p}\quad \text{or}\quad v^{\top}HH^{\top}v > c_2'\dfrac{n^{1+\tau}}{p}\right) < 4e^{-Cn},\]
for $H$ defined in Section~\ref{app:notation}. In particular, for $v=\beta$, a similar inequality holds for one side with a new $c_2'$ as
\[P\left(\beta^{\top}HH^{\top}\beta > c_2'\dfrac{n^{1+\tau}}{p}\right)<2e^{-Cn}.\]
\end{lemma}

\section{Proof of Theorem~\ref{thm:consistency}}
\label{app:consistency}
We first provide a sketch of the proof. By \eqref{eq:rpc} in Definition \ref{def:rpc} and Lemma \ref{lemma:rpcrewrite}, the ridge partial correlation can be rewritten as
\begin{equation}
    \label{eq:rpcrewrite}
    R_{i,\lambda} = -\dfrac{\hat{\beta}_i(\lambda)}{\left\{S_{\lambda}\xi_i + \hat{\beta}_i(\lambda)^2\right\}^{1/2}},
  \end{equation}
  where $\hat{\beta}_i(\lambda)$ is the $i$th element of the ridge regression estimator $\hat{\beta}(\lambda)$ and $\xi_i$ is the $i$th diagonal element of $(X^{\top}X + \lambda I)^{-1}$ as defined in Section~\ref{app:notation}.
  Thus, it follows that \[\dfrac{S_\lambda R_{i,\lambda}^2}{1-R_{i,\lambda}^2} = \dfrac{\hat{\beta}_i^2(\lambda)}{\xi_i}.\]
  By Theorem 3 in \cite{wang:leng:2016}, we have a separation on $|\hat{\beta}_i (\lambda)|$ for $i\in t$ and $i\notin t$. Therefore if we can show that $\lambda\xi_i\stackrel{p}{\to} 1$ uniformly, then we will have a separation on 
$S_\lambda R_{i,\lambda}^2/(1-R_{i,\lambda}^2)$ as stated in Theorem~\ref{thm:consistency}. Next, we provide a proof of Theorem~\ref{thm:consistency}.

\begin{proof}

From Section~\ref{app:notation}, since $X = VDU^{\top}\Sigma^{1/2}$, we have
\begin{equation}
\label{proof:ridgeproj}
    \begin{split}
        X^{\top}(XX^{\top}+\lambda I)^{-1}X &= \Sigma^{1/2}UDV^{\top}(VDU^{\top}\Sigma UDV^{\top}+\lambda I)^{-1}VDU^{\top}\Sigma^{1/2}\\
        &=\Sigma^{1/2}U(U^{\top}\Sigma U+\lambda D^{-2})^{-1}U^{\top}\Sigma^{1/2}\\
        &=\Sigma^{1/2}UA^{-1}(I+\lambda A^{-\top}D^{-2}A^{-1})^{-1}A^{-\top}U^{\top}\Sigma^{1/2}.
    \end{split}
\end{equation}

In order to use Taylor expansion on the inverse matrix in \eqref{proof:ridgeproj}, we need to consider the maximum eigenvalue of $A^{-\top}D^{-2}A^{-1}$. By Lemma \ref{lemma:sumofeigen}, 
\[\alpha_{\max}(A^{-\top}D^{-2}A^{-1})\leq \alpha_{\max}(D^{-2})\alpha_{\max}((AA^{\top})^{-1}) = \alpha_{\min}(D^2)^{-1}\alpha_{\min}(U^{\top}\Sigma U)^{-1}.\]
From the assumption A2, we have
\[P(p^{-1}\alpha_{\min}(D^2)<c_1^{-1})<e^{-C_1n}\]
for some $c_1>1$ and $C_1>0$. Also, by Lemma \ref{lemma:sumofeigen} and the assumption A3, we have
\[\alpha_{\min}(U^{\top}\Sigma U)\geq \alpha_{\min}(\Sigma)\alpha_{\min}(U^{\top}U)\geq c_4^{-1}n^{-\tau}.\]
Therefore, with probability greater than $1-e^{-C_1n}$, we have
\begin{equation}
    \label{proof:maxeigvalue}
    \alpha_{\max}(\lambda A^{-\top}D^{-2}A^{-1})\leq \dfrac{c_1c_4\lambda n^{\tau}}{p} < 1
\end{equation}
since by the assumption A4, $p\succ n^{1+\tau}$ and $\lambda \prec n^{1-(5/2)\tau-\kappa}$.

Thus, the Taylor expansion gives that
\begin{equation}
\label{proof:HH+M}
    \begin{split}
  &\Sigma^{1/2}UA^{-1}(I+\lambda A^{-\top}D^{-2}A^{-1})^{-1}A^{-\top}U^{\top}\Sigma^{1/2} \\
  &= \Sigma^{1/2}UA^{-1}\Bigg(I+\sum_{k=1}^\infty \lambda^k(A^{-\top}D^{-2}A^{-1})^k\Bigg)A^{-\top}U^{\top}\Sigma^{1/2}  \\
  &=HH^{\top}+M,
\end{split}
\end{equation}
where $H$ and $M$ are as defined in Appendix~\ref{app:notation}.

By Lemma \ref{lemma:sumofeigen}, we have 
\begin{equation}
\label{proof:maxwigenM}
\begin{split}
    &\alpha_{\max}(\Sigma^{1/2}UA^{-1}(A^{-\top}D^{-2}A^{-1})^kA^{-\top}U^{\top}\Sigma^{1/2})\\
  &\leq \alpha_{\max}(\Sigma^{1/2}UA^{-1}A^{-\top}U^{\top}\Sigma^{1/2})\alpha_{\max}(A^{-\top}D^{-2}A^{-1})^k\\
  &=\alpha_{\max}(HH^{\top})\alpha_{\max}(A^{-\top}D^{-2}A^{-1})^k\\
  &\leq \alpha_{\max}(A^{-\top}D^{-2}A^{-1})^k.
\end{split}
\end{equation}

Combining \eqref{proof:maxeigvalue}, \eqref{proof:maxwigenM} and Lemma \ref{lemma:sumofeigen}, we have an upper bound for the largest eigenvalue of $M$ given by
\begin{equation}
    \label{proof:maxeigvalueM}
    \alpha_{\max}(M)\leq\sum_{k=1}^\infty \lambda^k\alpha_{\max}(A^{-\top}D^{-2}A^{-1})^k\leq\dfrac{\lambda\alpha_{\max}(A^{-\top}D^{-2}A^{-1})}{1-\lambda\alpha_{\max}(A^{-\top}D^{-2}A^{-1})}.
\end{equation}

Since $g/(1-g)$ is an increasing function of $g$, from \eqref{proof:maxeigvalue}, \eqref{proof:maxeigvalueM} and the assumption A4, we have
\begin{equation}
\label{proof:maxeigenvalueMupperbd}
    \begin{split}
        P\left(\alpha_{\max}(M) > \dfrac{c_1c_4\lambda n^{\tau}}{p-c_1c_4\lambda n^{\tau}}\right) &< P\left(\dfrac{\lambda\alpha_{\max}(A^{-\top}D^{-2}A^{-1})}{1-\lambda\alpha_{\max}(A^{-\top}D^{-2}A^{-1})}>\dfrac{c_1c_4\lambda n^{\tau}}{p-c_1c_4\lambda n^{\tau}}\right)\\
        &=P\left(\alpha_{\max}(A^{-\top}D^{-2}A^{-1}) > \dfrac{c_1c_4n^{\tau}}{p}\right)<e^{-C_1n}.
    \end{split}
\end{equation}

By the Woodbury identity, \eqref{proof:ridgeproj}, and \eqref{proof:HH+M}, we have
\begin{equation}
\label{eq:xxt}
    (X^{\top}X + \lambda I)^{-1} = \dfrac{1}{\lambda}I-\dfrac{1}{\lambda}X^{\top}(XX^{\top}+\lambda I)^{-1}X = \dfrac{1}{\lambda}I-\dfrac{1}{\lambda}(HH^{\top}+M),
\end{equation}
and hence 
\begin{equation}
\label{eq:xtxii}
    \xi_i = \dfrac{1}{\lambda} - \dfrac{1}{\lambda}e_i^{\top}HH^{\top}e_i - \dfrac{1}{\lambda}e_i^{\top}Me_i,
  \end{equation}
  where $e_i = (0,\dots,0,1,0,\dots,0)^{\top}$ is the $i$th vector in the standard basis for $\mathcal{R}^p$.
By Lemma \ref{lemma:Hnormbound}, we have
\begin{equation}
    \label{eq:eiH}
    P\left(e_i^{\top}HH^{\top}e_i < c_2'\dfrac{n^{1+\tau}}{p}\right) > 1 - 4e^{-Cn}
\end{equation}
uniformly for all $i$. Also, by \eqref{proof:maxeigenvalueMupperbd}, we have 
\begin{equation}
\label{eq:eiM}
    P\left(e_i^{\top}Me_i > \dfrac{c_1c_4\lambda n^{\tau}}{p-c_1c_4\lambda n^{\tau}}\right)\leq P\left(\alpha_{\max}(M) > \dfrac{c_1c_4\lambda n^{\tau}}{p-c_1c_4\lambda n^{\tau}}\right) < e^{-C_1n}
\end{equation}
uniformly for all $i$.

Combining \eqref{eq:eiH} and \eqref{eq:eiM} and using Bonferroni's inequality, we have
\begin{equation}
\label{eq:maxi}
    P\left(\max_i(e_i^{\top}HH^{\top}e_i + e_i^{\top}Me_i) > c_2'\dfrac{n^{1+\tau}}{p} + \dfrac{c_1c_4\lambda n^{\tau}}{p-c_1c_4\lambda n^{\tau}}\right)
    \leq p(4e^{-Cn}+e^{-C_1n})\to 0
\end{equation}
as $n\to\infty$. From \eqref{eq:xtxii}, we have $\xi_{i} < 1/\lambda$. Thus, by \eqref{eq:xtxii} and \eqref{eq:maxi}, we have 
\begin{equation}
    \label{proof:maxlamxi}
    P\left(\max_{i}(|\lambda \xi_{i}-1|) > c_2'\dfrac{n^{1+\tau}}{p} + \dfrac{c_1c_4\lambda n^{\tau}}{p-c_1c_4\lambda n^{\tau}}\right) \leq p(4e^{-Cn}+e^{-C_1n}).
\end{equation}
By the assumption A4, 
\[\dfrac{n^{1+\tau}}{p} \to 0, \quad\text{and}\quad \dfrac{c_1c_4\lambda n^{\tau}}{p-c_1c_4\lambda n^{\tau}}\to 0\]
as $n\to\infty$, 
which means that $\lambda \xi_i\stackrel{p}{\to} 1$ uniformly.

We denote $c_2'n^{1+\tau}/p + c_1c_4\lambda n^{\tau}/(p-c_1c_4\lambda n^{\tau})$ by $\eta_{n}$ which is an infinitesimal.  
Let $\{\rho_{n}\}_{n \ge 1}$ be a sequence satisfying 
\begin{equation}
  \label{eq:defrho}
\dfrac{1}{\sqrt{1-\eta_n}}\left(\dfrac{\tilde{c}+\sigma\sqrt{2C_1c_1c_2'c_4}}{\sqrt{\log n}} + \dfrac{c_1c_4\sqrt{c'c_4}\lambda n^{5/2\tau + \kappa -1}}{1-c_1c_4\lambda n^{\tau}/p}\right)\dfrac{n^{1-\tau-\kappa}}{p}<\sqrt{\dfrac{\rho_{n}^2}{\lambda}} < \dfrac{c}{4}\dfrac{n^{1-\tau-\kappa}}{p},  
\end{equation}
where $c$ and $\tilde{c}$ are constants as in the proof of Theorem 3 in \cite{wang:leng:2016}.
Since 
\[\lim\limits_{n\to 0} \dfrac{1}{\sqrt{1-\eta_n}}\left(\dfrac{\tilde{c}+\sigma\sqrt{2C_1c_1c_2'c_4}}{\sqrt{\log n}} + \dfrac{c_1c_4\sqrt{c'c_4}\lambda n^{5/2\tau + \kappa -1}}{1-c_1c_4\lambda n^{\tau}/p}\right) = 0,\]
the $\rho_n$ in \eqref{eq:defrho} always exists.

Note that $1/R^2_{i,\lambda} = 1 + \xi_{i}\times S_{\lambda} / \hat{\beta}_i(\lambda)^2$. By the proof of Theorem 3 in \cite{wang:leng:2016} (see Supplementary of \cite{wang:leng:2016}), \eqref{proof:maxlamxi}, and the assumption A4, we have 
\begin{equation}
\label{proof:upperbound}
    \begin{split}
      P\left(\max_{i\notin t}\dfrac{S_\lambda R_{i,\lambda}^2}{1-R_{i,\lambda}^2} > \rho^2_{n}\right) &= P\left(\max_{i\notin t}\dfrac{\hat{\beta}_i^2(\lambda)}{\xi_i} > \rho_n^2\right)\\
      &\le P\left(\dfrac{\max_{i\notin t}\hat{\beta}_i^2(\lambda)}{\min_{i\notin t}\xi_i} > \rho_n^2\right)\\
      &= P\left(\max_{i\notin t}\hat{\beta}^2_i(\lambda) > \rho_n^2 \min_{i\notin t} \xi_{i}\right)\\
        &\leq P\left(\max_{i\notin t} \hat{\beta}^2_i(\lambda) > \dfrac{\rho^2(1-\eta_{n})}{\lambda}, \min_{i\notin t}\lambda \xi_{i} > 1-\eta_{n}\right) + P\left(\min_{i\notin t} \lambda \xi_{i} < 1-\eta_{n}\right)\\
        &\leq P\left(\max_{i\notin t}\hat{\beta}^2_i(\lambda) > (1-\eta_{n}) \dfrac{\rho^2_{n} }{\lambda}\right) + P\left(\max_{i\notin t}(1- \lambda \xi_{i}) > \eta_{n}\right)\\
        &\leq P\left(\max_{i\notin t}\hat{\beta}^2_i(\lambda) > \left[\left(\dfrac{\tilde{c}+\sigma\sqrt{2C_1c_1c_2'c_4}}{\sqrt{\log n}} + \dfrac{c_1c_4\sqrt{c'c_4}\lambda n^{5/2\tau + \kappa -1}}{1-c_1c_4\lambda n^{\tau}/p}\right)\dfrac{n^{1-\tau-\kappa}}{p}\right]^2\right)\\
        &+ P\left(\max_{i\notin t}(1- \lambda \xi_{i}) > \eta_{n}\right)\\
        &\leq O\bigg(\exp\bigg\{-C_1\dfrac{n^{1-5\tau-2\kappa-\nu}}{2\log(n)}\bigg\} \\
        &+\exp\bigg[1-\dfrac{1}{2}q\bigg(\dfrac{\sqrt{C_1}n^{1/2-2\tau-\kappa}}{\sqrt{\log(n)}}\bigg)\bigg]\bigg) + p(4e^{-Cn}+e^{-C_1n}).
    \end{split}
\end{equation}

Similarly, since $\xi_i < 1/\lambda$, we have
\begin{equation}
\label{proof:lowerbound}
\begin{split}
  P\left(\min_{i\in t}\dfrac{S_\lambda R_{i,\lambda}^2}{1-R_{i,\lambda}^2} < \rho^2_{n}\right) &= P\left(\min_{i\in t}\dfrac{\hat{\beta}_i^2(\lambda)}{\xi_i} <\rho_n^2\right)\\
  &= P\left(\dfrac{\min_{i\in t}\hat{\beta}_i^2(\lambda)}{\max_{i\in t} \xi_i} <\rho_n^2\right)\\
    &\leq P\left(\min_{i\in t}\hat{\beta}^2_i(\lambda) < \dfrac{\rho^2_{n} }{\lambda}\right)\\
    &< P\left(\min_{i\in t}\hat{\beta}^2_i(\lambda) < \left[\dfrac{c}{4}\dfrac{n^{1-\tau-\kappa}}{p}\right]^2\right)\\
    &<O\bigg(\exp\bigg\{-C_1\dfrac{n^{1-5\tau-2\kappa-\nu}}{2\log(n)}\bigg\} +
  \exp\bigg[1-\dfrac{1}{2}q\bigg(\dfrac{\sqrt{C_1}n^{1/2-2\tau-\kappa}}{\sqrt{\log(n)}}\bigg)\bigg]\bigg),
\end{split}
\end{equation}
where the last inequality follows from the proof of Theorem 3 in \cite{wang:leng:2016}.

By the assumption A4, $\log(p) \prec n^{1-2\kappa-5\tau}/[2\log (n)]$ and hence, 
\[p(4e^{-Cn}+e^{-C_1n})\prec \exp\bigg\{-C_1\dfrac{n^{1-5\tau-2\kappa-\nu}}{2\log(n)}\bigg\}.\]
This result combined with \eqref{proof:upperbound} and \eqref{proof:lowerbound} gives 
\begin{equation*}
    \begin{split}
        P\bigg(\min_{i\in t} \dfrac{S_\lambda R_{i,\lambda}^2}{1-R_{i,\lambda}^2}&> \rho^2_{n} > \max_{i\notin t}\dfrac{S_\lambda R_{i,\lambda}^2}{1-R_{i,\lambda}^2}\bigg)\\
        &= 1 - O\bigg(\exp\bigg\{-C_1\dfrac{n^{1-5\tau-2\kappa-\nu}}{2\log(n)}\bigg\} +
  \exp\bigg[1-\dfrac{1}{2}q\bigg(\dfrac{\sqrt{C_1}n^{1/2-2\tau-\kappa}}{\sqrt{\log(n)}}\bigg)\bigg]\bigg).
        \end{split}
\end{equation*}

Since $g/(1-g)$ is an increasing function of $g$, if we choose a submodel with size $K_n\geq |t|$, we will
have
\[P(t\subset \gamma^{K_n}) = 1- O\left[\exp\left\{-C_1\dfrac{n^{1-2\kappa-5\tau-\nu}}{2\log(n)}\right\} + \exp\left\{1-\dfrac{1}{2}q\left(\dfrac{\sqrt{C_1}n^{1/2-2\tau-\kappa}}{\sqrt{\log(n)}}\right)\right\}\right]\]
which completes the whole proof.
\end{proof}

\section{Additional simulation results}
\label{app:simu}
In this section, we present results from simulation settings corresponding to different $R^2$ values and error distributions. The results are reported as percentages.
\begin{center}
\begin{small}
\topcaption{\label{tab:addsimulationnorm} Additional simulation results for normally distributed errors}
\tablefirsthead{\hline\hline\multicolumn{1}{c}{Method}    & \multicolumn{12}{c}{Correlation Structure}\\}
\tablelasttail{\hline\hline}

\begin{supertabular}{rrr|rr|rr|rr|rr|rr|rr}
\hline\multicolumn{1}{l}{}                        & \multicolumn{12}{c}{Theoretical $R^2 = 0.1$}
                                                                                                                                                    \\\hline
                                                  & \multicolumn{2}{c|}{IID}      & \multicolumn{2}{c|}{Compound} & \multicolumn{2}{c|}{Group}    & \multicolumn{2}{c|}{AR}       & \multicolumn{2}{c|}{Factor}  & \multicolumn{2}{c|}{ExtrCor}  & \multicolumn{2}{c}{SpFactor}\\
                           & TPR            & CP         & TPR           & CP          & TPR            & CP         & TPR            & CP         & TPR           & CP         & TPR           & CP    &TPR      &CP      \\\hline
                           \multicolumn{1}{l}{RPC1}                            & 43.4 & 0 & 9.9 & 0 & 82.2 & 46 & 75.4 & 18 & 12.1 & 0 & 66.2 & 3 & 38.8 & 0 \\
                        \multicolumn{1}{l}{RPC2}    & 43.6 & 0 & 9.9 & 0 & 82.3 & 46 & 75.6 & 18 & 12.1 & 0 & 66.4 & 3 & 38.9 & 0 \\
                        \multicolumn{1}{l}{RPC3}    & 43.4 & 0 & 9.9 & 0 & 82.2 & 46 & 75.3 & 18 & 12.1 & 0 & 66.6 & 3 & 38.7 & 0 \\
                        \multicolumn{1}{l}{URPC}    & 43.6 & 0 & 9.9 & 0 & 82.3 & 46 & 75.6 & 18 & 12.2 & 0 & 67.0 & 3 & 38.9 & 0 \\
                        \multicolumn{1}{l}{UBITS}   & 45.6 & 0 & 17.3 & 0 & 87.8 & 64 & 69.4 & 5 & 20.1 & 0 & 16.3 & 0 & 55.3 & 1 \\
                        \multicolumn{1}{l}{HOLP}    & 43.4 & 0 & 9.9 & 0 & 82.4 & 47 & 75.3 & 18 & 12.1 & 0 & 59.2 & 1 & 39.1 & 0 \\
                        \multicolumn{1}{l}{FR}      & 15.4 & 0 & 6.2 & 0 & 21.0 & 0 & 21.0 & 0 & 7.9 & 0 & 9.1 & 0 & 9.3 & 0 \\
                        \multicolumn{1}{l}{SIS}     & 46.2 & 0 & 14.2 & 0 & 88.4 & 63 & 83.2 & 29 & 20.0 & 0 & 1.9 & 0 & 48.3 & 0 \\\hline
\multicolumn{1}{l}{}                        & \multicolumn{12}{c}{Theoretical $R^2 = 0.3$}                                                                                                                                                         \\\hline
                        & \multicolumn{2}{c|}{IID}      & \multicolumn{2}{c|}{Compound} & \multicolumn{2}{c|}{Group}    & \multicolumn{2}{c|}{AR}       & \multicolumn{2}{c|}{Factor}  & \multicolumn{2}{c|}{ExtrCor}  & \multicolumn{2}{c}{SparseFactor}\\
                           & TPR            & CP         & TPR           & CP          & TPR            & CP         & TPR            & CP         & TPR           & CP         & TPR           & CP    &TPR      &CP      \\\hline
\multicolumn{1}{l}{RPC1} & 89 & 33 & 24 & 0 & 100 & 100 & 100 & 98 & 18 & 0 & 92 & 46 & 63 & 2 \\
\multicolumn{1}{l}{RPC2} & 89 & 33 & 24 & 0 & 100 & 100 & 100 & 98 & 18 & 0 & 92 & 45 & 64 & 2 \\
\multicolumn{1}{l}{RPC3} & 89 & 33 & 24 & 0 & 100 & 100 & 100 & 98 & 18 & 0 & 92 & 46 & 63 & 2 \\
\multicolumn{1}{l}{URPC} & 89 & 33 & 24 & 0 & 100 & 100 & 100 & 98 & 18 & 0 & 92 & 47 & 64 & 2 \\
\multicolumn{1}{l}{UBITS} & 90 & 35 & 29 & 0 & 100 & 100 & 98 & 85 & 29 & 0 & 62 & 6 & 74 & 6 \\
\multicolumn{1}{l}{HOLP} & 89 & 33 & 24 & 0 & 100 & 100 & 100 & 98 & 18 & 0 & 90 & 37 & 64 & 2 \\
\multicolumn{1}{l}{FR} & 57 & 0 & 10 & 0 & 37 & 0 & 42 & 0 & 8 & 0 & 19 & 0 & 16 & 0 \\
\multicolumn{1}{l}{SIS} & 90 & 39 & 35 & 0 & 100 & 100 & 100 & 99 & 21 & 0 & 0 & 0 & 68 & 2 \\\hline
\multicolumn{1}{l}{}                        & \multicolumn{12}{c}{Theoretical $R^2 = 0.4$}         \\\hline
                        & \multicolumn{2}{c|}{IID}      & \multicolumn{2}{c|}{Compound} & \multicolumn{2}{c|}{Group}    & \multicolumn{2}{c|}{AR}       & \multicolumn{2}{c|}{Factor}  & \multicolumn{2}{c|}{ExtrCor}  & \multicolumn{2}{c}{SpFactor}\\
                           & TPR            & CP         & TPR           & CP          & TPR            & CP         & TPR            & CP         & TPR           & CP         & TPR           & CP    &TPR      &CP      \\\hline
\multicolumn{1}{l}{RPC1} & 96 & 69 & 32.4 & 0 & 100 & 100 & 99.9 & 99 & 36 & 0 & 97.4 & 80 & 69.0 & 1 \\
\multicolumn{1}{l}{RPC2} & 96 & 69 & 32.4 & 0 & 100 & 100 & 99.9 & 99 & 36.2 & 0 & 97.6 & 81 & 69.0 & 1 \\
\multicolumn{1}{l}{RPC3} & 96 & 69 & 32.3 & 0 & 100 & 100 & 99.9 & 99 & 36 & 0 & 97.3 & 79 & 69.0 & 1 \\
\multicolumn{1}{l}{URPC} & 96 & 69 & 32.6 & 0 & 100 & 100 & 99.9 & 99 & 36.2 & 0 & 97.6 & 81 & 69.0 & 1 \\
\multicolumn{1}{l}{UBITS} & 96.8 & 75 & 35.7 & 0 & 100 & 100 & 99.6 & 96 & 45.3 & 2 & 89.8 & 57 & 78.7 & 7 \\
\multicolumn{1}{l}{HOLP} & 96 & 69 & 32.4 & 0 & 100 & 100 & 99.9 & 99 & 36.4 & 0 & 95.6 & 67 & 69.1 & 1 \\
\multicolumn{1}{l}{FR} & 82.8 & 30 & 11.7 & 0 & 36.8 & 0 & 50.6 & 0 & 13.6 & 0 & 38.6 & 4 & 17.4 & 0 \\
\multicolumn{1}{l}{SIS} & 96.4 & 71 & 41.7 & 0 & 100.0 & 100 & 100.0 & 100 & 20.2 & 0 & 0.0 & 0 & 72.8 & 3 \\\hline
\end{supertabular}
\end{small}
\end{center}                                                                                                                                                    

\clearpage
\begin{center}
\begin{small}
\topcaption{\label{tab:simulationtdist} Simulation results for $t$ distributed errors}
\tablefirsthead{\hline\hline\multicolumn{1}{c}{Method}    & \multicolumn{12}{c}{Correlation Structure}\\}
\tablelasttail{\hline\hline}
\begin{supertabular}{rrr|rr|rr|rr|rr|rr|rr}
\hline
\multicolumn{1}{l}{}                        & \multicolumn{12}{c}{Theoretical $R^2 = 0.1$} \\
\hline
                        & \multicolumn{2}{c|}{IID}      & \multicolumn{2}{c|}{Compound} & \multicolumn{2}{c|}{Group}    & \multicolumn{2}{c|}{AR}       & \multicolumn{2}{c|}{Factor}  & \multicolumn{2}{c|}{ExtrmCor}  & \multicolumn{2}{c}{SparseFactor} \\
                           & TPR            & CP         & TPR           & CP          & TPR            & CP         & TPR            & CP         & TPR           & CP         & TPR           & CP    & TPR & CP \\
\hline
\multicolumn{1}{l}{RPC1}      & 43.6 & 0 & 10.6 & 0 & 80.6 & 40  & 78.7 & 19 & 10.7 & 0 & 67.1 & 2 & 41.1 & 0    \\
\multicolumn{1}{l}{RPC2}      & 43.6 & 0 & 10.6 & 0 & 80.6 & 40  & 78.7 & 19 & 10.7 & 0 & 67.2 & 2 & 41.1 & 0    \\
\multicolumn{1}{l}{RPC3}      & 43.6 & 0 & 10.4 & 0 & 80.4 & 40  & 78.6 & 19 & 10.7 & 0 & 67.0 & 2 & 41.1 & 0    \\
\multicolumn{1}{l}{URPC}      & 43.6 & 0 & 10.6 & 0 & 80.6 & 40  & 78.7 & 19 & 10.8 & 0 & 67.3 & 2 & 41.1 & 0    \\
\multicolumn{1}{l}{UBITS}     & 46.1 & 0 & 17.8 & 0 & 87.3 & 62  & 71.9 & 7  & 17.8 & 0 & 16.8 & 0 & 57.2 & 2 \\
\multicolumn{1}{l}{HOLP}      & 43.6 & 0 & 10.6 & 0 & 80.9 & 42  & 78.6 & 19 & 10.9 & 0 & 59.4 & 0 & 41.2 & 0 \\
\multicolumn{1}{l}{FR}        & 16.2 & 0 & 6.3  & 0 & 21.9 & 0   & 22.7 & 0  & 7.1  & 0 & 8.9  & 0 & 9.0  & 0 \\
\multicolumn{1}{l}{SIS}       & 45.6 & 0 & 16.1 & 0 & 87.7 & 68  & 82.6 & 26 & 19.0 & 0 & 4.2  & 0 & 50.3 & 1 \\
\hline
\multicolumn{1}{l}{}                        & \multicolumn{12}{c}{Theoretical $R^2 = 0.2$} \\
\hline
                        & \multicolumn{2}{c|}{IID}      & \multicolumn{2}{c|}{Compound} & \multicolumn{2}{c|}{Group}    & \multicolumn{2}{c|}{AR}       & \multicolumn{2}{c|}{Factor}  & \multicolumn{2}{c|}{ExtrmCor}  & \multicolumn{2}{c}{SparseFactor} \\
                           & TPR            & CP         & TPR           & CP          & TPR            & CP         & TPR            & CP         & TPR           & CP         & TPR           & CP    & TPR & CP \\
\hline
\multicolumn{1}{l}{RPC1}     & 74.3 & 6 & 17.7 & 0 & 99.1 & 96 & 96.1 & 70  & 14.2 & 0 & 82.3 & 20 & 54.4 & 0    \\
\multicolumn{1}{l}{RPC2}     & 74.3 & 6 & 17.7 & 0 & 99.1 & 96 & 96.1 & 70  & 14.7 & 0 & 82.4 & 20 & 54.5 & 0    \\
\multicolumn{1}{l}{RPC3}     & 74.3 & 6 & 17.7 & 0 & 99.1 & 96 & 96.1 & 70  & 14.1 & 0 & 82.3 & 20 & 54.4 & 0    \\
\multicolumn{1}{l}{URPC}     & 74.3 & 6 & 17.7 & 0 & 99.1 & 96 & 96.1 & 70  & 14.8 & 0 & 82.6 & 20 & 54.5 & 0    \\
\multicolumn{1}{l}{UBITS}    & 75.1 & 6 & 24.3 & 0 & 99.9 & 99 & 92.3 & 45  & 22.3 & 0 & 41.8 & 1 & 68.8 & 1 \\
\multicolumn{1}{l}{HOLP}     & 74.3 & 6 & 17.7 & 0 & 99.1 & 96 & 96.1 & 70  & 14.4 & 0 & 78.2 & 8 & 54.5 & 0    \\
\multicolumn{1}{l}{FR}       & 34.7 & 0 & 8.0  & 0 & 34.1 & 0   & 33.1 & 0  & 7.9  & 0 & 12.8 & 0 & 13.0 & 0    \\
\multicolumn{1}{l}{SIS}      & 76.3 & 8 & 24.9 & 0 & 99.6 & 97 & 98.2 & 84  & 20.9 & 0 & 0.3  & 0 & 61.4 & 0 \\
\hline
\multicolumn{1}{l}{}                        & \multicolumn{12}{c}{Theoretical $R^2 = 0.3$} \\
\hline
                        & \multicolumn{2}{c|}{IID}      & \multicolumn{2}{c|}{Compound} & \multicolumn{2}{c|}{Group}    & \multicolumn{2}{c|}{AR}       & \multicolumn{2}{c|}{Factor}  & \multicolumn{2}{c|}{ExtrmCor}  & \multicolumn{2}{c}{SparseFactor} \\
                           & TPR            & CP         & TPR           & CP          & TPR            & CP         & TPR            & CP         & TPR           & CP         & TPR           & CP    & TPR & CP \\
\hline
\multicolumn{1}{l}{RPC1}    & 88.8 & 30 & 22.7 & 0 & 100.0 & 100 & 98.8 & 89  & 23.7 & 0 & 91.2 & 44 & 63.2 & 1 \\
\multicolumn{1}{l}{RPC2}    & 88.8 & 30 & 22.7 & 0 & 100.0 & 100 & 98.8 & 89  & 23.8 & 0 & 91.1 & 45 & 63.2 & 1 \\
\multicolumn{1}{l}{RPC3}    & 88.8 & 30 & 22.7 & 0 & 100.0 & 100 & 98.8 & 89  & 23.7 & 0 & 91.3 & 45 & 63.2 & 1 \\
\multicolumn{1}{l}{URPC}    & 88.8 & 30 & 22.7 & 0 & 100.0 & 100 & 98.8 & 89  & 23.8 & 0 & 91.6 & 47 & 63.2 & 1 \\
\multicolumn{1}{l}{UBITS}   & 89.2 & 35 & 28.4 & 0 & 100.0 & 100 & 98.1 & 85  & 31.6 & 0 & 62.6 & 5 & 74.8 & 5 \\
\multicolumn{1}{l}{HOLP}    & 88.8 & 30 & 22.7 & 0 & 100.0 & 100 & 98.8 & 89  & 24.0 & 0 & 88.8 & 35 & 63.2 & 1 \\
\multicolumn{1}{l}{FR}      & 57.3 & 3  & 7.6  & 0 & 36.3 & 0   & 42.9 & 0  & 10.3 & 0 & 17.9 & 0 & 15.8 & 0 \\
\multicolumn{1}{l}{SIS}     & 89.2 & 33 & 30.8 & 0 & 100.0 & 100 & 98.9 & 90  & 18.8 & 0 & 0.0  & 0 & 68.4 & 3 \\
\hline
\multicolumn{1}{l}{}                        & \multicolumn{12}{c}{Theoretical $R^2 = 0.4$} \\
\hline
                        & \multicolumn{2}{c|}{IID}      & \multicolumn{2}{c|}{Compound} & \multicolumn{2}{c|}{Group}    & \multicolumn{2}{c|}{AR}       & \multicolumn{2}{c|}{Factor}  & \multicolumn{2}{c|}{ExtrmCor}  & \multicolumn{2}{c}{SparseFactor} \\
                           & TPR            & CP         & TPR           & CP          & TPR            & CP         & TPR            & CP         & TPR           & CP         & TPR           & CP    & TPR & CP \\
\hline
\multicolumn{1}{l}{RPC1}    & 95.6 & 65 & 35.2 & 0 & 100.0 & 100 & 100.0 & 100 & 29.3 & 0  & 96.8 & 74 & 62.9 & 2 \\
\multicolumn{1}{l}{RPC2}    & 95.6 & 65 & 35.3 & 0 & 100.0 & 100 & 100.0 & 100 & 29.1 & 0  & 96.8 & 74 & 63.0 & 2 \\
\multicolumn{1}{l}{RPC3}    & 95.6 & 65 & 35.2 & 0 & 100.0 & 100 & 100.0 & 100 & 28.9 & 0  & 96.9 & 75 & 62.9 & 2 \\
\multicolumn{1}{l}{URPC}    & 95.6 & 65 & 35.4 & 0 & 100.0 & 100 & 100.0 & 100 & 29.4 & 0  & 96.9 & 75 & 63.0 & 2 \\
\multicolumn{1}{l}{UBITS}   & 95.6 & 68 & 38.9 & 0 & 100.0 & 100 & 99.8 & 98 & 37.8 & 1  & 88.2 & 52 & 74.1 & 5 \\
\multicolumn{1}{l}{HOLP}    & 95.6 & 65 & 35.1 & 0 & 100.0 & 100 & 100.0 & 100 & 29.7 & 0  & 94.7 & 61 & 63.2 & 2 \\
\multicolumn{1}{l}{FR}      & 83.1 & 33 & 12.7 & 0 & 37.4 & 0  & 51.1 & 0  & 12.9 & 0  & 37.0 & 4  & 16.4 & 0 \\
\multicolumn{1}{l}{SIS}     & 95.6 & 65 & 41.1 & 0 & 100.0 & 100 & 100.0 & 100 & 22.1 & 0  & 0.0  & 0  & 66.1 & 4 \\
\hline
\end{supertabular}

\end{small}
\end{center}                                                                                                 
\clearpage
\begin{center}
\begin{small}
\topcaption{\label{tab:simulationexp} Simulation results for exponentially distributed errors}
\tablefirsthead{\hline\hline\multicolumn{1}{c}{Method}    & \multicolumn{12}{c}{Correlation Structure}\\}
\tablelasttail{\hline\hline}

\begin{supertabular}{rrr|rr|rr|rr|rr|rr|rr}
\hline\multicolumn{1}{l}{}                        & \multicolumn{12}{c}{Theoretical $R^2 = 0.1$}
                                                \\\hline
                        & \multicolumn{2}{c|}{IID}      & \multicolumn{2}{c|}{Compound} & \multicolumn{2}{c|}{Group}    & \multicolumn{2}{c|}{AR}       & \multicolumn{2}{c|}{Factor}  & \multicolumn{2}{c|}{ExtrmCor}  &
                        \multicolumn{2}{c}{SparseFactor}\\
                           & TPR            & CP         & TPR           & CP          & TPR            & CP         & TPR            & CP         & TPR           & CP         & TPR           & CP    &TPR      &CP      \\\hline
\multicolumn{1}{l}{RPC1}      & 45.6 & 0    & 10.7 & 0 & 82.3 & 44 & 77.2 & 18 & 9.3 & 0 & 67.1 & 2 & 41.4 & 0    \\
                         
\multicolumn{1}{l}{RPC2}      & 45.6 & 0    & 10.7 & 0 & 82.3 & 44 & 77.2 & 18 & 9.3 & 0 & 67.2 & 2 & 41.4 & 0    \\
                        
\multicolumn{1}{l}{RPC3}      & 45.6 & 0    & 10.7 & 0 & 82.2 & 44 & 77.2 & 18 & 9.3 & 0 & 66.9 & 2 & 41.3 & 0    \\
                        
\multicolumn{1}{l}{URPC}      & 45.6 & 0    & 10.7 & 0 & 82.3 & 44 & 77.2 & 18 & 9.4 & 0 & 67.4 & 2 & 41.4 & 0    \\
                              
\multicolumn{1}{l}{UBITS}  & 48.3 & 1 & 18.0  & 0 & 86.7 & 66 & 73.3 & 12 & 17.2 & 0 & 18.8 & 0    & 58.1 & 1 \\
                    
\multicolumn{1}{l}{HOLP}     & 45.7 & 0    & 10.7 & 0 & 82.3 & 44 & 77.2 & 18 & 9.4 & 0 & 61.0  & 0    & 41.6 & 0    \\
                        
\multicolumn{1}{l}{FR}        & 15.1 & 0    & 6.3 & 0 & 22.4 & 0    & 23.0  & 0    & 6.9 & 0 & 8.7 & 0    & 8.9 & 0    \\
                       
\multicolumn{1}{l}{SIS}                             & 47.0  & 0    & 17.4 & 0 & 88.8 & 60  & 82.8 & 25 & 17.6 & 0 & 1.3 & 0    & 51.0  & 0   \\\hline

\end{supertabular}

\end{small}
\end{center}                                                                                                                                






\end{document}